\documentclass[12pt]{article}

\usepackage{amsmath,amsfonts,amssymb,amsthm,amscd}
\usepackage{epsfig}

\newtheorem{theorem}{Theorem}[section]             
            
\newtheorem{lemma}[theorem]{Lemma}
\newtheorem{corollary}[theorem]{Corollary}
\newtheorem*{conjecture}{Conjecture}

\newtheorem*{question}{Question}
\newtheorem{obs}[theorem]{Observation}
\newtheorem{claim}[theorem]{Claim}

\usepackage{hyperref}

\def\inst#1{$^{#1}$}

\long\def\onefigure#1#2{
\begin{figure*}[tbh]
\begin{center}
\includegraphics{#1}
\end{center}
\caption{#2}
\end{figure*}
} 

\newcommand{\myfig}[2]  
{\onefigure{{g-#1.eps}}{\label{f:#1} #2} }

\begin{document}

\title{Universal Sets for Straight-Line Embeddings of Bicolored Graphs\thanks{
Research was supported by the project 1M0545 of the Ministry of Education of the Czech Republic 
and by the grant SVV-2010-261313 (Discrete Methods and Algorithms).
Viola M\'{e}sz\'{a}ros was also partially supported by OTKA grant K76099
and by European project IST-FET AEOLUS.
Josef Cibulka and Rudolf Stola\v{r} were also 
supported by the Czech Science Foundation under the contract no.\ 201/09/H057.
Josef Cibulka, Jan Kyn\v{c}l and Pavel Valtr were also supported by the Grant Agency 
of the Charles University, GAUK 52410.
\newline 
Part of the research was conducted during the Special Semester on Discrete
and Computational Geometry at \'Ecole Polytechnique F\'ed\'erale de Lausanne, organized and
supported by the CIB (Centre Interfacultaire Bernoulli) and the SNSF (Swiss
National Science Foundation).
} \thanks{A preliminary version appeared in 
proceedings of Graph Drawing 2008~\cite{gdversion}.}
} 

\author{Josef Cibulka\inst{1}, Jan Kyn\v{c}l\inst{2}, Viola M\'{e}sz\'{a}ros\inst{2,3}, \\
Rudolf Stola\v{r}\inst{1} and Pavel Valtr\inst{2}
} 

\date{}

\maketitle

\begin{center}
{\footnotesize
\inst{1} 
Department of Applied Mathematics, \\
Charles University, Faculty of Mathematics and Physics, \\
Malostransk\'e n\'am.~25, 118~ 00 Prague, Czech Republic; \\ 
\texttt{cibulka@kam.mff.cuni.cz, ruda@kam.mff.cuni.cz} 
\\\ \\
\inst{2}
Department of Applied Mathematics and Institute for Theoretical Computer Science, \\
Charles University, Faculty of Mathematics and Physics, \\
Malostransk\'e n\'am.~25, 118~ 00 Prague, Czech Republic; \\
\texttt{kyncl@kam.mff.cuni.cz}
\\\ \\
\inst{3}
Bolyai Institute, University of Szeged, \\
Aradi v\'ertan\'uk tere 1, 6720 Szeged, Hungary; \\
\texttt{viola@math.u-szeged.hu}
}
\end{center}  


\begin{abstract}
A set $S$ of $n$ points is \emph{$2$-color universal} for a graph $G$ on $n$ vertices if for every proper 
$2$-coloring of $G$ and for every $2$-coloring of $S$ with the same sizes of color classes as $G$ has, 
$G$ is straight-line embeddable on $S$.

We show that the so-called double chain is $2$-color universal for paths if each of the two
chains contains at least one fifth of all the points, but not if one of the chains is more than 
approximately $28$ times longer than the other.

A $2$-coloring of $G$ is \emph{equitable} if the sizes of the color classes differ by at most $1$.
A bipartite graph is \emph{equitable} if it admits an equitable proper coloring.
We study the case when $S$ is the double-chain with chain sizes differing 
by at most $1$ and $G$ is an equitable bipartite graph. We prove that this $S$ is not $2$-color universal 
if $G$ is not a forest of caterpillars and that it is $2$-color universal for equitable caterpillars 
with at most one half non-leaf vertices. 
We also show that if this $S$ is equitably $2$-colored, then equitably properly $2$-colored 
forests of stars can be embedded on it.

\end{abstract}


\maketitle

\section{Introduction}
\subsection{Previous Results}
It is frequently asked in geometric graph theory whether
a given graph $G$ can be drawn without edge crossings on a given planar point set $S$ 
under some additional constraints on the drawing. In this paper, we always assume 
$|V(G)| = |S|$.

One possibility is to prescribe a fixed position for each vertex of $G$.
If the edges are allowed to be arbitrary curves, then we can obtain a planar drawing 
of an arbitrary graph $G$ by moving vertices from an arbitrary planar drawing of $G$ 
to the given points. Pach and Wenger showed~\cite{pachwenger} that every planar $G$ with 
prescribed vertex positions can be drawn so that each edge is a piecewise-linear curve with $O(n)$ 
bends and that this bound is tight even if $G$ is a path.

In another setting, we are only given the graph $G$ and the set of points $S$, but we
are allowed to choose the mapping between $V(G)$ and $S$.
Kaufmann and Wiese~\cite{kaufmannwiese} showed that two bends per edge are then always enough 
and for some graphs necessary. An \emph{outerplanar graph} is a graph admitting a planar drawing 
where one face contains all vertices. By a result of Gritzmann et al.~\cite{gritzmann91}, 
outerplanar graphs are exactly the graphs with a straight-line planar drawing on 
an arbitrary set of points in general position.

In this paper we are dealing with a combination of these two versions. The vertices and
points are colored with two colors and each vertex has to be placed on a point of the 
same color. An obvious necessary condition to find such a drawing is that each color class
has the same number of vertices as points. 
We then say that the $2$-coloring of $S$ is \emph{compatible} with the $2$-coloring of $V(G)$.

A \emph{caterpillar} is a tree in which the non-leaf vertices induce a path. 
The coloring of $V(G)$ is \emph{proper} if it doesn't create any monochromatic edge.
It is known that drawing some bicolored planar graphs on some bicolored point sets requires 
at least $\Omega(n)$ bends per edge~\cite{giacomograph},
but caterpillars can be drawn with two bends per edge~\cite{giacomograph} and one bend per
edge is enough for paths~\cite{giacomopath} and properly colored caterpillars~\cite{giacomograph}.

We restrict our attention to the proper $2$-colorings of a bipartite graph $G$. Then 
the question of embeddability of $G$ on a bicolored point set is very similar to finding a non-crossing
copy of $G$ in a complete geometric graph from which we removed edges of the two complete subgraphs
on points of the two color classes. 
The only difference is that in the latter case, we can swap colors on some connected components of $G$.
A related question was posed by Micha Perles on DIMACS Workshop on Geometric Graph Theory in 2002.
He asked what is the maximum number $h(n)$ such that if we remove arbitrary $h(n)$ edges  
from a complete geometric graph on an arbitrary set of $n$ points in general position, we
can still find a non-crossing Hamiltonian path.
\v{C}ern\'y et al.~\cite{cerny07} showed that $h(n) = \Omega(\sqrt{n})$
and also that it is safe to remove the edges of an arbitrary complete graph
on $\Omega(\sqrt{n})$ vertices and that this bound is asymptotically optimal. 
Aichholzer et al.~\cite{aichholzer10} summarize history and results of this type also for
graphs different from the path.

It is thus impossible to find a non-crossing Hamiltonian alternating (that is, properly colored) path (\emph{NHAP} for short) 
on some bicolored point sets. Kaneko et al.~\cite{kanekokanosuzuki} proved that the smallest such point set 
has $16$ points if we allow only even number $n$ of points and $13$ for arbitrary $n$. 

Several sufficient conditions are known under which a bicolored point set admits an NHAP. An NHAP exists whenever
the two color classes are separable by a line~\cite{abellanas99} or if one of them is composed of 
the points of the convex hull of $S$~\cite{abellanas99}.

The result on sets with color classes separable by a line readily implies that any $2$-colored 
set $S$ with each color class of size $n/2$ admits a non-crossing alternating path (NAP) on at least $n/2$ 
points of $S$. It is an open problem if this lower bound can be improved to $n/2+f(n)$, where $f(n)$
is unbounded (see also Chapter 9.7 of the book~\cite{brassmoserpach}). 
On the other hand, there are such $2$-colored sets admitting no NAP of length more than $\approx2n/3$
\cite{abellanas03,kynclpt}.
This upper bound is proved for certain colorings of points in convex position. The above
general lower bound $n/2$ can be slightly improved for sets in convex position~\cite{kynclpt, hajnalmeszaros}, 
the best bound is currently $n/2+\Omega(\sqrt{n})$ by Hajnal and M\'esz\'aros~\cite{hajnalmeszaros}.

The main result of this paper is that some point sets contain an NHAP for any equitable 
$2$-coloring of their points. We call such point sets \emph{$2$-color universal} for a path.

See the survey of Kaneko and Kano~\cite{kanekokano} for more results on embedding graphs on bicolored point sets.
One of the results mentioned in the survey is the possibility to embed some graphs with a fixed $2$-coloring on
an arbitrary compatibly $2$-colored point set. Let $G$ be a forest of stars with centers colored black and 
leaves white and let $S$ be a $2$-colored point set. If we map the centers of stars arbitrarily and then 
we map the leaves so that the sum of lengths of edges is minimized, then no two edges cross.

Previously, a different notion of universality was considered in the context of embedding colored graphs 
on colored point sets. A $k$-colored set $S$ of $n$ points is \emph{universal} for a class $\mathcal G$ 
of graphs if every (not necessarily proper) coloring of vertices of $G\in \mathcal G$ on $n$ vertices 
admits an embedding on $S$. Brandes et al.~\cite{brandes10+} find, for example, universal $k$-colored 
point sets for the class of caterpillars for every $k \leq 3$.

\subsection{Our Results}
A {\em convex\/} or a {\em concave chain\/} is a finite set of points in the plane
lying on the graph of a strictly convex or a strictly concave function, respectively.
A {\em double-chain} $(C_1,C_2)$ consists of a convex chain $C_1$ and a concave chain
$C_2$ such that each point of $C_2$ lies strictly below every line determined
by $C_1$ and similarly, each point of $C_1$ lies strictly above every line determined
by $C_2$ (see Fig.~\ref{f:double-chain}).
Double-chains were first considered in \cite{garcia}. 

The size of a chain $C_i$ is the number $|C_i|$ of its points.
Note that we allow different sizes of the chains $C_1$ and $C_2$. If the sizes $|C_1|$, $|C_2|$
of the chains differ by at most $1$, we say that the double-chain is \emph{balanced}.

We consider only $2$-colorings and we use {\em black\/} and {\em white\/} as the colors.
A $2$-coloring of a set $S$ of $n$ points in the plane is \emph{compatible} with a $2$-coloring of a graph $G$ 
on $n$ vertices if the number of black points of $S$ is the same as the number of black vertices in $G$. 
This implies the equality of numbers of white points and white vertices as well. 

A graph $G$ with $2$-colored vertices is \emph{embeddable} on a $2$-colored point set $S$ if the vertices 
of $G$ can be mapped to the points of $S$ so that the colors match and no two edges cross if they are drawn as 
straight-line segments.

A set $S$ of points is \emph{$2$-color universal} for a bipartite graph $G$ if for every proper $2$-coloring
of $G$ and for every compatible $2$-coloring of $S$, $G$ is embeddable on $S$.

If the properly colored path on $n$ vertices, $P_n$, can be embedded on a $2$-colored set $S$ of $n$ points, 
then we say that $S$ has an NHAP (non-crossing Hamiltonian alternating path).

A $2$-coloring of a set $S$ of $n$ points is \emph{equitable} if it is compatible with some 
proper $2$-coloring of $P_n$, that is, if the sizes of the two color classes differ by at most $1$.

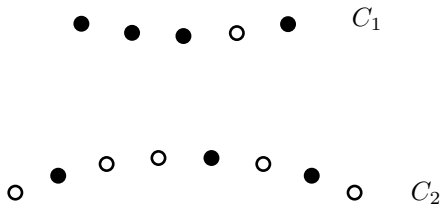
\myfig{double-chain}{an equitable $2$-coloring of a double-chain $(C_1, C_2)$}

Section~\ref{s:mainthm} contains the proof of the following theorem.

\begin{theorem}\label{t:main}
Let $(C_1,C_2)$ be a double-chain satisfying $|C_i| \ge 1/5 (|C_1|+|C_2|)$ for $i=1,2$.
Then $(C_1,C_2)$ is $2$-color universal for $P_n$, where $n=|C_1|+|C_2|$. Moreover, an NHAP on an equitably 
colored $(C_1,C_2)$ can be found in linear time.
\end{theorem}

Note that this doesn't always hold if we don't require the coloring of the path to be proper. For example, 
if we color first two vertices of $P_4$ black and the other two white, then it cannot be embedded on the 
double-chain with both chains of size $2$ if the two black points are the top-left one and the bottom-right one.

In Section~\ref{s:nopath}, we show that double-chains with highly unbalanced sizes of chains
do not admit an NHAP for some equitable $2$-colorings.

\begin{theorem} \label{t:nopath}
Let $(C_1,C_2)$ be a double-chain, and let $C_1$ be periodic with the following period of length 16:
2 black, 4 white, 6 black and 4 white points.
If $|C_1| \ge 28(|C_2|+1)$, then $(C_1,C_2)$ has no NHAP.
\end{theorem}

An \emph{equitable coloring} of a graph is a coloring where the sizes of any two color classes differ by at most $1$.
A bipartite graph is \emph{equitable} if it admits a proper equitable $2$-coloring.

Section~\ref{sec:emb} mainly studies other graphs for which the balanced double-chain is $2$-color universal.

\begin{theorem}
\label{thm:eqonbal}
The balanced double-chain is $2$-color universal for equitable caterpillars with at least as many leaves 
as non-leaf vertices. 
\par If a forest of stars is $2$-colored equitably and properly, then it can be embedded on every 
compatibly $2$-colored balanced double-chain.
\par If the balanced double chain is $2$-color universal for an equitable bipartite graph $G$, then
$G$ is a forest of caterpillars.
\end{theorem}

We also present examples of equitable bipartite planar graphs for which no set of points is $2$-color universal.

\section{Proof of Theorem~\ref{t:main}}\label{s:mainthm}

The main idea of our proof is to cover the chains $C_i$ by a special type of
pairwise non-crossing paths, so called
hedgehogs, and then to connect these hedgehogs into an NHAP by adding some
edges between $C_1$ and $C_2$.

\subsection{Notation Used in the Proof}

For $i=1,2$, let $b_i$ be the number of black points of $C_i$ and
let $w_i:=|C_i|- b_i$ denote the number of white points of
$C_i$. 

Since the coloring is equitable, we may assume that 
$b_1 \ge w_1$ and $w_2 \ge b_2$.
Then black is {\em the major color of} $C_1$ and
{\em the minor color of} $C_2$, and white is {\em the major color
of} $C_2$ and {\em the minor color of} $C_1$. Points in the major
color, i.e., black points on $C_1$ and white points on $C_2$, are
called {\em major points}. Points in the minor color are called
{\em minor points}.

Points on each $C_i$ are linearly ordered according to the $x$-coordinate.
An {\em interval\/} of $C_i$ is a sequence of consecutive points of $C_i$.
An {\em inner point\/} of an interval $I$ is any point of $I$ which is
neither the leftmost nor the rightmost point of $I$.

\myfig{hedgehog}{a hedgehog in $C_1$}

A {\em body} $D$ is a non-empty interval of a chain $C_i$ $(i=1,2)$
such that all inner points of $D$ are major.
If the leftmost point of $D$ is minor, then we call it a {\em head\/}
of $D$. Otherwise $D$ has no head. If the rightmost point of $D$
is minor, then we call it a {\em tail\/} of $D$. Otherwise $D$ has no tail.
If a body consists of just one minor point, this point is both the head
and the tail.

Bodies are of the following four types. A {\em $00$-body\/} is a body with
no head and no tail. A {\em $11$-body\/} is a body with both head
and tail. The bodies of remaining two types have exactly one endpoint major
and the other one minor. We will call the body a {\em $10$-body\/} or
a {\em $01$-body\/} if the minor endpoint is a head or a tail, respectively.

Let $D$ be a body on $C_i$. A {\em hedgehog (built on the body
$D\subseteq C_i$)\/} is a non-crossing alternating path $H$
with vertices in $C_i$ satisfying
the following three conditions: (1) $H$ contains all points of $D$, (2) $H$
contains no major points outside of $D$, (3) the endpoints of $H$ are
the first and the last point of $D$.
A hedgehog built on an $\alpha\beta$-body is an {\em $\alpha\beta$-hedgehog\/}
($\alpha, \beta=0,1$). If a hedgehog $H$ is built on a body $D$, then $D$ is
{\em the body of $H$\/} and the points of $H$ that do not lie in $D$ are
{\em spines}. Note that each spine is a minor point. All possible types of 
hedgehogs can be seen on Fig.~\ref{f:htypes} (for better lucidity, we will draw
hedgehogs with bodies on a horizontal line and spines indicated only by a
``peak'' from now on).

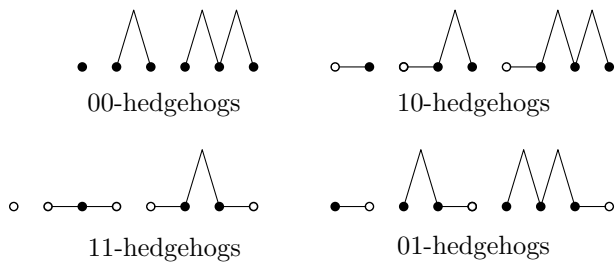
\myfig{htypes}{types of hedgehogs (sketch)}

On each $C_i$, maximal intervals containing only major points are called {\em runs}.
Clearly, runs form a partition of major points.
For $i=1,2$, let $r_i$ denote the number of runs in $C_i$.

\subsection{Proof in the Even Case}

Throughout this subsection, $(C_1,C_2)$ denotes a double-chain with
$|C_1|+|C_2|$ even. Since the coloring is equitable, we have $b_1+b_2=w_1+w_2$.
Set $$\Delta:=b_1-w_1=w_2-b_2.$$

First we give a lemma characterizing collections of bodies
on a chain $C_i$ that are bodies of some pairwise non-crossing
hedgehogs covering the whole chain $C_i$.

\begin{lemma}\label{l:hedgehogs}
Let $i\in\{1,2\}$. Let all major points of $C_i$ be covered by a set
${\cal D}$ of pairwise disjoint bodies. Then the bodies of ${\cal D}$ are the bodies
of some pairwise non-crossing hedgehogs covering the whole $C_i$ if and only if
$\Delta=d_{00}-d_{11}$, where $d_{\alpha\alpha}$ is the number
of $\alpha\alpha$-bodies in ${\cal D}$.
\end{lemma}

\begin{proof}
An $\alpha\beta$-hedgehog containing $t$ major points contains
$(t-1)+\alpha+\beta$ minor points. It follows that the equality
$\Delta=d_{00}-d_{11}$ is necessary for the existence of a covering
of $C_i$ by disjoint hedgehogs built on the bodies of ${\cal D}$.

Suppose now that $\Delta=d_{00}-d_{11}$. Let $F$ be the set of minor points
on $C_i$ that lie in no body of ${\cal D}$, and let $M$ be the set of
the mid-points of
straight-line segments connecting pairs of consecutive major points
lying in the same body. It is easily checked that $|F|=|M|$.
Clearly $F\cup M$ is a convex or a concave chain. Now it is easy to prove
that there is a non-crossing perfect matching formed by $|F|=|M|$ straight-line
segments between $F$ and $M$ (for the proof, take any segment connecting
a point of $F$ with a neighboring point of $M$, remove the two points,
and continue by induction); see Fig.~\ref{f:spines}.

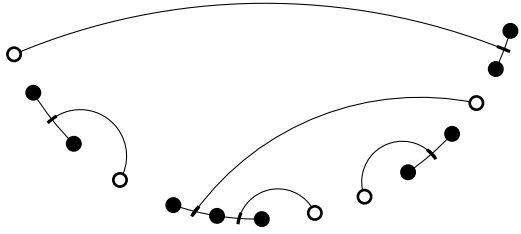
\myfig{spines}{a non-crossing matching of minor points and midpoints (in $C_1)$}

If $f\in F$ is connected to a point $m\in M$ in the matching,
then $f$ will be a spine with edges going from it to those two major
points that determined $m$. Obviously, these spines and edges define
non-crossing hedgehogs with bodies in ${\cal D}$ and with all the required properties.
\end{proof}

\bigskip

The following three lemmas and their proofs show how to construct an NHAP
in some special cases.

\begin{lemma}\label{l:deltalarge}
If $\Delta\ge\max\{r_1,r_2\}$ then $(C_1,C_2)$ has an NHAP.
\end{lemma}

\begin{proof}
Let $i\in\{1,2\}$.
Since $r_i\le \Delta\le \max(b_i,w_i)$, the runs in $C_i$ may be partitioned
into $\Delta$ $00$-bodies. By Lemma~\ref{l:hedgehogs}, these $00$-bodies
may be extended to pairwise non-crossing hedgehogs covering $C_i$.
This gives us $2\Delta$ hedgehogs on the double-chain.
They may be connected into an NHAP by $2\Delta-1$ edges between the chains
in the way shown in Fig.~\ref{f:lemma2}.
\end{proof}

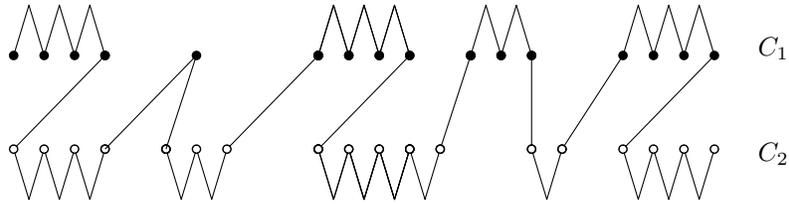
\myfig{lemma2}{$00$-hedgehogs connected to an NHAP}

\begin{lemma}\label{l:equalruns}
If $r_1=r_2$ then $(C_1,C_2)$ has an NHAP.
\end{lemma}

\begin{proof}
Set $r:=r_1=r_2$.
If $r\le\Delta$ then we may apply Lemma~\ref{l:deltalarge}.
Thus, let $r>\Delta$.

Suppose first that $\Delta\ge1$. We cover each run on each $C_i$ by a single body
whose type is as follows. On $C_1$ we take $\Delta$ $00$-bodies
followed by $(r-\Delta)$ $10$-bodies. On $C_2$ we take (from left to right)
$(\Delta-1)$ $00$-bodies, $(r-\Delta)$ $01$-bodies, and one $00$-body.
By Lemma~\ref{l:hedgehogs}, the $r$ bodies on each $C_i$ can be extended
to hedgehogs covering $C_i$. Altogether we obtain $2r$ hedgehogs. They can be connected
to an NHAP by $2r-1$ edges between $C_1$ and $C_2$ (see Fig.~\ref{f:lemma3}).

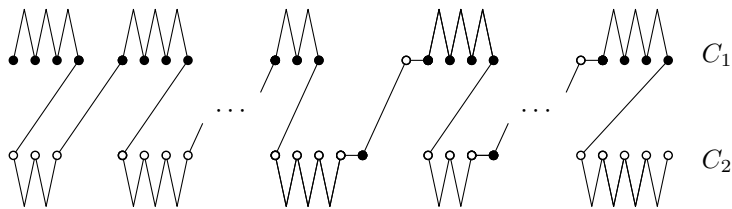
\myfig{lemma3}{an NHAP in the case $r_1=r_2>\Delta \geq 1$}

Suppose now that $\Delta=0$. We add one auxiliary major point on each $C_i$
as follows. On $C_1$, the auxiliary point extends the leftmost run on the left.
On $C_2$, the auxiliary point extends the rightmost run on the right.
This does not change the number of runs and increases $\Delta$ to $1$.
Thus, we may proceed as above. The NHAP obtained has the two auxiliary points
on its ends. We may remove the auxiliary points from the path, obtaining
an NHAP for $(C_1,C_2)$.
\end{proof}
\bigskip

A {\em singleton} $s\in C_i$ is an inner point of $C_i$ ($i=1,2$) such that
its two neighbors on $C_i$ are colored differently from $s$.

\begin{lemma}\label{l:nosingletons}
Suppose that $C_1$ has no singletons and $C_2$ can be covered by $r_1-1$
pairwise disjoint hedgehogs.
Then $(C_1,C_2)$ has an NHAP.
\end{lemma}

\begin{proof}
For simplicity of notation, set $r:=r_1$.
We denote the $r-1$ hedgehogs on $C_2$ by $P_1,P_2,\dots,P_{r-1}$
in the left-to-right order in which the bodies of these hedgehogs appear
on $C_2$. For technical reasons, we enlarge
the leftmost run of $C_1$ from the left by an auxiliary major point $\sigma$.

Our goal is to find $r$ hedgehogs $H_1,H_2,\dots,H_{r}$
on $C_1\cup\{\sigma\}$ such that they may be connected with the hedgehogs
$P_1,P_2,\dots,P_{r-1}$ into an NHAP. For each $j=1,\dots,r$, the body 
of the hedgehog $H_j$ will be denoted by $D_j$. For each $j=1,\dots,r$,
$D_j$ covers the $j$-th run of $C_1\cup\{\sigma\}$ (in the left-to-right order).
We now finish the definition of the bodies $D_j$ by specifying
for each $D_j$ if it has a head and/or a tail.
The body $D_1$ is without head. For $j>1$, $D_j$ has a head if and
only if $P_{j-1}$ has a tail. The last body $D_r$ is without tail and 
$D_j,j<r$, has a tail if and only if $P_j$ has a head.

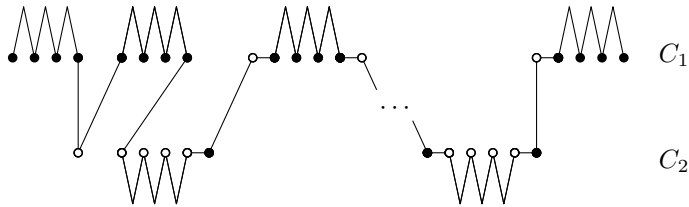
\myfig{lemma4}{an NHAP in the case of no singletons on $C_1$}

It follows from Lemma~\ref{l:hedgehogs} that we
may add or remove some minor points on $C_1\cup\{\sigma\}$
so that $D_1,\dots,D_r$ can then be extended to pairwise
non-crossing hedgehogs $H_1,\dots,H_r$ covering the ``new'' $C_1$.
More precisely, there is a double-chain $(C_1',C_2)$ such that
$D_1,\dots,D_r$ can be extended to pairwise
non-crossing hedgehogs $H_1,\dots,H_r$ covering $C_1'$, where
either $C_1'=C_1\cup\{\sigma\}$ or $C_1'$ is obtained from $C_1\cup\{\sigma\}$
by adding some minor (white) points on the left of $C_1\cup\{\sigma\}$ (say)
or $C_1'$ is obtained from $C_1\cup\{\sigma\}$ by removal of some minor (white)
points lying in none of the bodies $D_1,\dots,D_r$.
Then the concatenation $H_1P_1H_2P_2\cdots H_{r-1}P_{r-1}H_r$
shown in Fig.~\ref{f:lemma4} gives an NHAP on $(C_1',C_2)$.
This NHAP starts with the point $\sigma$. Removal
of $\sigma$ from it gives an NHAP $P$
for the double-chain $(C_1'\setminus\{\sigma\},C_2)$. The endpoints of $P$
have different colors. Thus, $P$ covers the same number
of black and white points. Black points on $P$ are the $(|C_1|+|C_2|)/2$
black points of $(C_1,C_2)$. Thus, $P$ covers exactly $|C_1|+|C_2|$
points. It follows that $|C_1'\setminus\{\sigma\}|=|C_1|$ and thus
$C_1'\setminus\{\sigma\}=C_1$.
The path $P$ is an NHAP on the double-chain $(C_1,C_2)$.
\end{proof}

\bigskip

The following lemma will be used to find a covering needed in 
Lemma~\ref{l:nosingletons}.

\begin{lemma}\label{l:paths}
Suppose that $|C_i|\geq k$, $r_i\leq k$ and $\Delta\leq k$ for some $i\in
\{1, 2\}$ and for some integer $k$. Then $C_i$ can be covered by $k$ pairwise
disjoint hedgehogs.
\end{lemma}

\begin{proof}
The idea of the proof is to start with the set $\mathcal{D}$ of $|C_i|$
bodies, each of them being a single point, and then gradually decrease the 
number of bodies in $\mathcal{D}$ by joining some of the bodies together.
We see that $\Delta=d_{00}-d_{11}$, where $d_{\alpha\alpha}$ is the number of
$\alpha\alpha$-bodies in $\mathcal{D}$. If we join two neighboring $00$-bodies to one
$00$-body and withdraw a single-point $11$-body from $\mathcal{D}$ (to let the
minor point become a spine) at the same time, the difference between the number
of $00$-bodies and the number of $11$-bodies remains the same and
$|\mathcal{D}|$ decreases by two. We can reduce $|\mathcal{D}|$ by one
while preserving the difference $d_{00}-d_{11}$
by joining a $00$-body with a neighboring single-point
$11$-body into a $01$- or a $10$-body. Similarly we can join a $01$- or
a $10$-body with a neighboring (from the proper side) single-point $11$-body
into a new $11$-body to decrease $|\mathcal{D}|$ by one as well. When we are
joining two $00$-bodies, we choose the single-point $11$-body to remove in such
a way to keep as many single-point $11$-bodies adjacent to $00$-bodies as
possible. This guarantees that we can use up to $r_i$ of them for heads and
tails.

We start with joining neighboring $00$-bodies and we do this as long as
$|\mathcal{D}|>k+1$ and $d_{00}>r_i$. Note that by the assumption 
$\Delta\leq k$, we will have enough single-point $11$-bodies to do that.
When we end, one of the following conditions holds: $|\mathcal{D}|=k$,
$|\mathcal{D}|=k+1$ or $d_{00}=r_i$. In the first case we are done. If
$|\mathcal{D}|=k+1$, we just add one head or one tail (we can do this since
$d_{00}+d_{11}=|\mathcal{D}|=k+1 \geq d_{00}-d_{11}+1$, which implies
$d_{11}>0$). If $d_{00}=r_i$, then each run is covered by just one $00$-body.
We need to add $|\mathcal{D}|-k$ heads and tails. We have enough
single-point $11$-bodies to do that since 
$d_{11}=|\mathcal{D}|-d_{00}=|\mathcal{D}|-r_i \geq |\mathcal{D}|-k$.
On the other hand, $r_i-d_{11}=\Delta \geq 0$, so the number of heads and
tails needed is at most $r_i$. Therefore, all the single-point $11$-bodies are
adjacent to $00$-bodies and we can use them to form heads and tails.

In all cases we get a set $\mathcal{D}$ of $k$ bodies. Now we can apply
Lemma~\ref{l:hedgehogs} to obtain $k$ pairwise disjoint hedgehogs covering
$C_i$.
\end{proof}

\bigskip

By a {\em contraction\/} we mean removing a singleton with both its
neighbors and putting a point of the color of its neighbors in its place
instead. It is easy to verify that if there is an NHAP in the new
double-chain obtained by this contraction, it can be expanded to an NHAP
in the original double-chain.

Now we can prove our main theorem in the even case.

\bigskip

\begin{pfoftmain}
Without loss of generality we may assume that $r_1 \geq r_2$. In the case
$r_1=r_2$, we get an NHAP by Lemma~\ref{l:equalruns}. In the case
$\Delta\ge r_1$, we get an NHAP by Lemma~\ref{l:deltalarge}. Therefore,
the only case left is $r_1>r_2$, $r_1>\Delta$.

If there is a singleton on $C_1$, we make a contraction of it. By this we
decrease $r_1$ by one and both $r_2$ and $\Delta$ remain unchanged. If
now $r_1=r_2$ or $r_1=\Delta$, we again get an NHAP, otherwise
we keep making contractions until one of the previous cases appears or there
are no more singletons to contract.

If there is no more singleton to contract on $C_1$ and still $r_1>r_2$ and
$r_1>\Delta$, we try to cover $C_2$ by $r_1-1$ pairwise disjoint paths.
Before the contractions, $|C_2| \geq |C_1|/4$ did hold and by the
contractions we could just decrease $|C_1|$, therefore it still holds.

All the maximal intervals on the chain $C_1$ (with possible exception of the
first and the last one) have now length at least two, which implies that
$r_1 \leq |C_1|/4 +1$.
Hence $|C_2| \geq |C_1|/4 \geq r_1 -1$, so we can create $r_1-1$
pairwise disjoint hedgehogs covering $C_2$ using Lemma~\ref{l:paths}.
Then we apply Lemma~\ref{l:nosingletons} and expand the NHAP obtained
by Lemma~\ref{l:nosingletons} to an NHAP on the original double-chain.

There is a straightforward linear-time algorithm for finding an NHAP
on $(C_1,C_2)$ based on the above proof.
\qed
\end{pfoftmain}



\subsection{Proof in the Odd Case}\label{s:odd}
In this subsection we prove Theorem \ref{t:main} for the case when $|C_1|+|C_2|$
is odd. We set $\Delta=w_2-b_2$ and proceed similarly as in the even case.
On several places in the proof we will add one auxiliary point
$\omega$ to get the even case (its color will be chosen to equalize
the numbers of black and white points). We will be able to apply one of the
Lemmas \ref{l:deltalarge}--\ref{l:nosingletons} to obtain an NHAP.
The point $\omega$ will be at some end of the NHAP and by removing
$\omega$ we obtain an NHAP for $(C_1, C_2)$.

Without loss of generality we may assume that $r_1 \geq r_2$. In the case
$r_1=r_2$, we add an auxiliary major point $\omega$, which is placed either as
the left neighbor of the leftmost major point on
$C_1$ or as the right neighbor of the rightmost major point on $C_2$. Then
we get an NHAP by Lemma \ref{l:equalruns} and the removal of $\omega$ gives
us an NHAP for $(C_1, C_2)$. 

In the case $\Delta\ge r_1$, we add an auxiliary point $\omega$ to the same place
and we get an NHAP by Lemma \ref{l:deltalarge}. Again, the removal of
$\omega$ gives us an NHAP for $(C_1, C_2)$. 

Now, the only case left is $r_1>r_2$, $r_1>\Delta$.
If there are any singletons on $C_1$, we make the contractions exactly the same
way as in the proof of the even case. If Lemma \ref{l:deltalarge} or
\ref{l:equalruns} needs to be applied, we again add an auxiliary point
$\omega$ and proceed as above.

If there is no more singleton to contract on $C_1$ and still $r_1>r_2$ and
$r_1>\Delta$, we have $|C_2| \geq |C_1|/4 \geq r_1 -1$ as in the proof of
the even case and we can use Lemma
\ref{l:paths} to get $r_1-1$ pairwise disjoint hedgehogs covering $C_2$. Now we need to consider two cases: (1) If $b_1 + b_2 > w_1 + w_2$, then we find
an NHAP for $(C_1, C_2)$ in the same way as in the proof of Lemma
\ref{l:nosingletons}, except we do not add the auxiliary point $\sigma$.
(2) If $b_1 + b_2 < w_1 + w_2$, we add an auxiliary point $\omega$ as the
right neighbor of the rightmost major point on $C_1$. The number $r_1$ didn't
change so Lemma \ref{l:nosingletons} gives us an NHAP. Again, the removal
of $\omega$ gives us an NHAP for $(C_1, C_2)$.

There is a straightforward linear-time algorithm for finding an NHAP
on $(C_1,C_2)$ based on the above proof.
\qed

\section{Unbalanced Double-Chains with no NHAP}\label{s:nopath}
In this section we prove Theorem~\ref{t:nopath}.
Let $(C_1,C_2)$ be a double-chain whose points are colored by an equitable 2-coloring, and let $C_1$ be periodic with the following period: 2 black, 4 white, 6 black and 4 white points. 
Let $|C_1| \geq 28(|C_2|+1)$. We want to show that $(C_1,C_2)$ has no NHAP.

Suppose on the contrary that $(C_1,C_2)$ has an NHAP.
Let $P_1, P_2,\dots, P_t$ denote the maximal subpaths of the NHAP containing only points of $C_1$. Since between every two consecutive paths $P_i$, $P_j$ in the NHAP there is at least one point of $C_2$, we have $t\leq |C_2|+1$. In the following we think of $C_1$ as of a cyclic sequence of points on the circle. Note that we get more intervals in this way. Theorem~\ref{t:nopath} now directly follows from the following theorem.

\begin{theorem}\label{t:circle_covering}
Let $C_1$ be a set of points on a circle periodically $2$-colored with the following period of length $16$:
$2$ black, $4$ white, $6$ black and $4$ white points.
Suppose that all points of $C_1$ are covered by a set of $t$ non-crossing alternating and pairwise disjoint paths $P_1, P_2,\dots, P_t$. Then $t > |C_1|/28$.
\end{theorem}

\begin{proof}
Each maximal interval spanned by a path $P_i$ on the circle is called a {\em base}. Let $b(P_i)$ denote the number of bases of $P_i$. A path with one base only is called a {\em leaf}. We consider the following special types of edges in the paths. {\em Long edges\/} connect points that belong to different bases. {\em Short edges\/} connect consecutive points on $C_1$. Note that short edges cannot be adjacent to each other. A maximal subpath of a path $P_i$ spanning two subintervals of two different bases and consisting of long edges only is called a {\em zig-zag}. A path is {\em separated\/} if all of its edges can be crossed by a line. Note that each zig-zag is a separated path. A maximal separated subpath of $P_i$ that contains an endpoint of $P_i$ and spans one interval only is a {\em rainbow\/}. We find all the zig-zags and rainbows in each $P_i$, $i=1,2,\dots, t$. Note that two zig-zags, or a zig-zag and a rainbow, are either disjoint or share an endpoint.
A {\em branch\/} is a maximal subpath of $P_i$ that spans two intervals and is induced by a union of zig-zags. 

For each path $P_i$ that is not a leaf construct the following graph $G_i$. The vertices of $G_i$ are the bases of $P_i$. We add an edge between two vertices for each branch that connects the corresponding bases.
If $G_i$ has a cycle (including the case of a ``$2$-cycle''), then one of the corresponding branches consists of a single edge that lies on the convex hull of $P_i$. We delete such an edge from $P_i$ and don't call it a branch anymore. By deleting a corresponding edge from each cycle of $G_i$ we obtain a graph $G'_i$, which is a spanning tree of $G_i$. The {\em branch graph} $G'$ is a union of all graphs $G'_i$.
 
Let ${\cal L }$ denote the set of leaves and ${\cal B }$ the set of branches. Let ${\cal P}$ = $\{P_1, P_2,\dots,$ $P_t\}$.

\begin{obs}\label{o:branches}
The branch graph $G'$ is a forest with components $G'_i$. Therefore, $$|{\cal B}|= \sum_{i, P_i \notin {\cal L}} (b(P_i)-1).$$ \qed
\end{obs}

The branches and rainbows in $P_i$ do not necessarily cover all the points of $P_i$. Each point that is not covered is adjacent to a deleted long edge and to a short edge that connects this point to a branch or a rainbow. It follows that between two consecutive branches (and between a rainbow and the nearest branch) there are at most two uncovered points, that are endpoints of a common deleted edge. By an easy case analysis it can be shown that this upper bound can be achieved only if one of the nearest branches consists of a single zig-zag.

In the rest of the paper, a {\em run\/} will be a maximal monochromatic interval of any color.
In the following we will count the runs that are spanned by the paths $P_i$. The {\em weight\/} of a path $P$, $w(P)$, is the number of runs spanned by $P$. If $P$ spans a whole run, it adds one unit to $w(P)$. If $P$ partially spans a run, it adds half a unit to $w(P)$.

\begin{obs}\label{o:mbranch}
The weight of a zig-zag or a rainbow is at most $1.5$. A branch consists of at most two zig-zags, hence it weights at most three units. \qed
\end{obs}

\begin{lemma} \label{l:path}
A path $P_i$ that is not a leaf weights at most $3.5k+3.5$ units where $k$ is the number of branches in $P_i$.
\end{lemma}

\begin{proof}
According to the above discussion, for each pair of uncovered points that are adjacent on $P_i$ we can join one of them to the adjacent branch consisting of a single zig-zag. To each such branch we join at most two uncovered points, hence its weight increases by at most one unit to at most $2.5$ units. The number of the remaining uncovered points is at most $k+1$. Therefore, $w(P_i) \le 3k + 3 + 0.5 \cdot (k+1) = 3.5k+3.5$.
\end{proof}

\begin{lemma}\label{l:leaf}
A leaf weights at most $3.5$ units.
\end{lemma}

\begin{proof}
Let $L$ be a leaf spanning at least two points. Consider the interval spanned by $L$. Cut this interval out of $C_1$ and glue its endpoints together to form a circle. Take a line $l$ that crosses the first and the last edge of $L$. Note that the line $l$ doesn't separate any of the runs. Exactly one of the arcs determined by $l$ contains the gluing point $\gamma$. 

Each of the ending edges of $L$ belongs to a rainbow, all of whose edges cross $l$. It follows that if $L$ has only one rainbow, then this rainbow covers the whole leaf $L$ and $w(L) \le 1.5$.
Otherwise $L$ has exactly two rainbows, $R_1$ and $R_2$. We show that $R_1$ and $R_2$ cover all edges of $L$ that cross the line $l$. Suppose there is an edge $s$ in $L$ that crosses $l$ and does not belong to any of the rainbows $R_1$, $R_2$. Then one of these rainbows, say $R_1$, is separated from $\gamma$ by $s$. Then the edge of $L$ that is the second nearest to $R_1$ also has the same property as the edge $s$. This would imply that $R_1$ spans two whole runs, a contradiction. It follows that all the edges of $L$ that are not covered by the rainbows are consecutive and connect adjacent points on the circle. There are at most three such edges; at most one connecting the points adjacent to $\gamma$, the rest of them being short on $C_1$. But this upper bound of three cannot be achieved since it would force both rainbows to span two whole runs. Therefore, there are at most two edges and hence at most one point in $L$ uncovered by the rainbows. The lemma follows.
\end{proof}

\begin{lemma}\label{l:leaves}
$|{\cal L}| \geq \sum_{i, P_i \notin {\cal L}}(b(P_i)-2) + 2$.
\end{lemma}

\begin{proof}
The number of runs in $C_1$ is at least $4$. By Lemma~\ref{l:leaf}, if all the paths $P_i$ are leaves, then at least $2$ of them are needed to cover $C_1$ and the lemma follows. 

If not all the paths are leaves, we order the paths so that all the leaves come at the end of the ordering. 
The path $P_1$ spans $b(P_1)$ bases. Shrink these bases to points. These points divide the circle into $b(P_1)$ arcs each of which contains at least one leaf. If $P_2$ is not a leaf then continue. The path $P_2$ spans $b(P_2)$ intervals on one of the previous arcs. Shrink them to points. These points divide the arc into $b(P_2)+1$ subarcs. At least $b(P_2)-1$ of them contain leaves. This increased the number of leaves by at least $b(P_2)-2$. The case of $P_i$, $i>2$, is similar to $P_2$. The lemma follows by induction.
\end{proof}

\begin{corollary} 
$|{\cal B}|\leq |{\cal P}|-2.$
\end{corollary}

\begin{proof}
Combining Lemma~\ref{l:leaves} and Observation~\ref{o:branches} we get the following:
$$ |{\cal B}| = \sum_{i, P_i \notin {\cal L}} (b(P_i)-1)= \sum_{i, P_i \notin {\cal L}} (b(P_i)-2)+ |{\cal P}|-|{\cal L}| + 2 - 2 \leq |{\cal P}|-2.$$
\end{proof}
\bigskip

Now we are in position to finish the proof of Theorem~\ref{t:circle_covering}.
If the whole $C_1$ is covered by the paths $P_i$, then $\sum_{i=1}^t w(P_i) \geq |C_1|/4$. Therefore,

$$|C_1| \leq 4\cdot(3.5|{\cal B}|+3.5(|{\cal P}|-|{\cal L}|)+ 3.5|{\cal L}|) < 4\cdot7|{\cal P}| = 28 |{\cal P}|.$$
\end{proof}

\section{Embedding equitable bipartite graphs}
\label{sec:emb}
\subsection{Embedding on balanced double-chains}
We already know that the balanced double-chain is $2$-color universal for the path $P_n$. 
In this subsection, we further study the class of graphs for which the balanced double-chain 
is $2$-color universal. The three lemmas of this subsection prove the three claims of 
Theorem~\ref{thm:eqonbal}.

\begin{lemma}
\label{lem:forcat}
If the balanced double-chain is $2$-color universal for an equitable bipartite graph $G$ 
then $G$ is a forest of caterpillars.
\end{lemma}
\begin{proof}
Let $K^+_{1,3}$ be the $3$-star with subdivided edges (see Fig.~\ref{f:k13}).
A connected graph is a caterpillar if and only if it contains no cycle and no $K^+_{1,3}$ as a subgraph.

We will color all points of one chain white and points of the other chain black so that the resulting 
coloring is compatible with the $2$-coloring of $G$. We assume for contradiction that $G$ can be 
embedded on it and that it contains either a cycle or $K^+_{1,3}$.

As the double-chain has monochromatic chains, all the edges connect the two chains.
Because the embedding has no edge crossings, we can consider the leftmost edge of the cycle 
and let $u$ and $v$ be its endvertices. Then the two edges of the cycle incident to exactly one of 
$u$ and $v$ cross.

We now assume that $K^+_{1,3}$ can be embedded on the double-chain and let the color of its root vertex 
be white. Let $\omega_1$ be 
the white point where the root of $K^+_{1,3}$ is mapped and let $\beta_1$, $\beta_2$, $\beta_3$ be 
(from left to right) the three black points where the middle vertices are mapped. Then $\beta_2$ 
is connected by an edge to some white leaf vertex of $K^+_{1,3}$, but this edge is crossed either by the 
segment $\omega_1 \beta_1$ or by $\omega_1 \beta_3$. See Fig.~\ref{f:k13}.
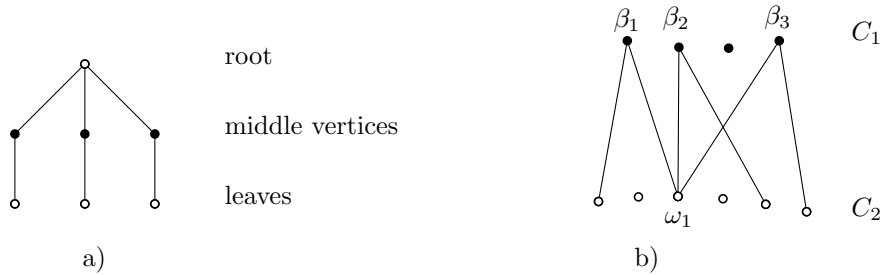
\myfig{k13}{a) $K^+_{1,3}$ and b) impossibility of its embedding on the double chain with monochromatic chains}
\end{proof}

The \emph{central path} in a caterpillar is the set of non-leaf vertices. 

\begin{lemma}
\label{lem:shortcat}
If an equitable bipartite graph $G$ on $n$ vertices is a caterpillar with at most $\lfloor n/2 \rfloor$ 
vertices on the central path, then the balanced double-chain is $2$-color universal for $G$.
\end{lemma}
\begin{proof}
Let $b_i$ be the number of black points on the chain $C_i$ and $w_i$ the number of white points on $C_i$.
Since the coloring is equitable, we may assume that $b_1 \ge w_1$ and $w_2 \ge b_2$.
Then black is {\em the major color} of $C_1$ and {\em the minor color} of $C_2$, and white is 
{\em the major color} of $C_2$ and {\em the minor color} of $C_1$. Points in the major
color are called {\em major points}. Points in the minor color are called {\em minor points}.

Observe that the number of minor points on each chain is at most the number of leaves of $G$ of that color. 
Let $G'$ be the graph with $b_1$ black and $w_1$ white vertices obtained from $G$ by removing some 
$b_2$ black and $w_1$ white leaves. 

In the first phase, we embed $G'$ on the set of major points of 
the two chains. We take the vertices of the central path of $G'$ starting from one of its ends. 
A vertex $v$ of the central path is placed on the leftmost unused major point on the chain 
where the color of $v$ is major. The leaves of $v$ in $G'$ are then successively placed on the leftmost unused 
major points of the other chain.

In the second phase, we map all the leaves removed in the first phase on minor points.
In the same greedy way as in the proof of Lemma~\ref{l:hedgehogs}, we keep connecting the closest pair
of an unused white point of $C_1$ and a black point of the central path that still misses at least 
one leaf. The same is done on $C_2$.

This guarantees that no crossing appears and that every vertex is mapped to some point. See 
Fig.~\ref{f:caterpillar}.
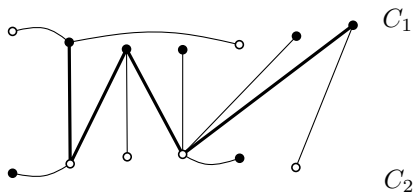
\myfig{caterpillar}{embedding a caterpillar on a balanced double-chain; the bold edges form the central path}
\end{proof}

\begin{lemma}
If a forest of stars $G$ is $2$-colored equitably and properly, then $G$ can be embedded on every 
compatibly $2$-colored balanced double-chain.
\end{lemma}
\begin{proof}
We take some fixed proper equitable $2$-coloring of $G$. 

We show that by adding edges to $G$, we are able to create a properly $2$-colored
caterpillar on the set of all vertices of $G$ and with at most $\lfloor n/2 \rfloor$ non-leaf vertices.
By Lemma~\ref{lem:shortcat} this caterpillar can be embedded on every compatibly $2$-colored balanced 
double-chain and thus $G$ can be embedded.

The cases when $n \leq 3$ and when $G$ has no edges are trivial.

For every $i\geq 3$, let $k_i$ ($h_i$) be the number of stars on $i$ vertices and with black (white) central vertex.
In case of $2$-vertex components of $G$, we cannot distinguish the 
central vertex and we let $n_2$ be their number. We also let $n_1$ be the total number of $1$-vertex 
components of $G$ as it is not necessary for the proof to count black and white ones separately.

We assume without loss of generality, that at least half of the stars on at least three vertices have black central vertex.
We start connecting the central vertices of stars on at least three vertices into an alternating path starting 
with a black vertex. At some point we run out of stars with white central vertex. We then use the stars on 
two vertices as stars with 
white central vertex. If we run out of stars with black central vertex, we use every second star on two vertices as 
a star with black central vertex. Otherwise we run out of stars on two vertices. Then we start connecting 
each of the remaining stars on at least three vertices by an edge between one of its white leaves and the last 
black vertex on the path.

The resulting graph is composed of a connected graph $T$ and all $1$-vertex components of $G$. The graph $T$ is 
a properly colored caterpillar and the created path is its central path $P$. 
Since $G$ has some edges, $P$ is not empty.

If $P$ contains only one vertex $v$, we pick one of its leaves, $u$, and connect every $1$-vertex component of $G$
either to $u$ or to $v$, depending on its color. The central path then has $2$ vertices, which is at most 
$\lfloor n/2 \rfloor$.

If $P$ has at least two vertices, we connect every $1$-vertex component of $G$ to a vertex of 
the other color on the central path. 

See Fig.~\ref{f:starstocat}.

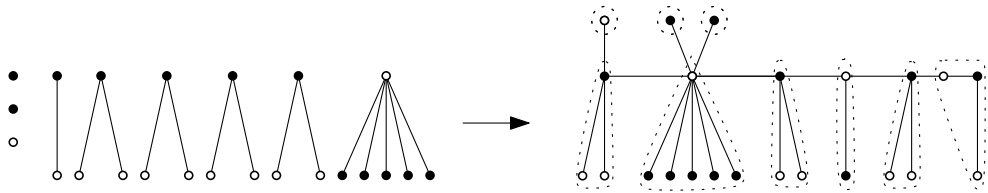
\myfig{starstocat}{connecting stars to form a caterpillar}

It remains to show that the central path is not too long. The total number of vertices is
\begin{equation} \label{eq:vercnt}
n = n_1+2n_2+\sum_{i=3}^{n}i(k_i+h_i).
\end{equation}

If $\sum_{i=3}^{n}k_i \leq n_2 + \sum_{i=3}^{n}h_i$, then every vertex of the central path has at least one leaf
and thus the caterpillar has at most $\lfloor n/2 \rfloor$ non-leaf vertices. 

Otherwise, the central path starts and ends with a black vertex and the black vertices of the central path are 
exactly the black centers of stars on at least three vertices. The central path thus has $2\sum_{i=3}^{n}k_i-1$ 
vertices.

Because the $2$-coloring of $G$ is equitable, the number of black vertices of $G$ is at least 
$\lfloor n/2 \rfloor$ and thus
\[
\sum_{i=3}^{n}k_i + \sum_{i=3}^{n}(i-1)h_i + n_2 + n_1 \geq \left\lfloor \frac{n}{2} \right\rfloor.
\]

At most $\lfloor n/2 \rfloor + 1$ vertices of $G$ are white, which leads to
\[
\left\lfloor \frac{n}{2} \right\rfloor + 1 \geq \sum_{i=3}^{n}h_i + \sum_{i=3}^{n}(i-1)k_i + n_2.
\]

Putting the two inequalities together gives us
\begin{equation}
\sum_{i=3}^{n}(i-2)h_i + n_1 \geq \sum_{i=3}^{n}(i-2)k_i - 1. \label{eq:balgraph}
\end{equation}

The number of vertices of the central path is at most $\lfloor n/2 \rfloor$, because
\begin{align*}
2 \left(2\sum_{i=3}^{n}k_i-1 \right) &\leq \sum_{i=3}^{n}2(i-1)k_i-2 \\
& < \sum_{i=3}^{n}ik_i + \sum_{i=3}^{n}(i-2)k_i - 1 \\
& \leq \sum_{i=3}^{n}ik_i + \sum_{i=3}^{n}(i-2)h_i + n_1 \qquad \textrm{by Eq.~\ref{eq:balgraph}}\\
& \leq \sum_{i=3}^{n}i(k_i+h_i) + n_1 \\
& \leq n  \qquad \textrm{by Eq.~\ref{eq:vercnt}.}
\end{align*}

\end{proof}

\subsection{Open problems}

It seems plausible that the balanced double-chain is $2$-color universal for all equitable forests
of caterpillars.

\begin{conjecture}
The balanced double-chain is $2$-color universal for an equitable bipartite graph $G$ if and only if
$G$ is a forest of caterpillars.
\end{conjecture}

Some graphs for which the balanced double-chain is not $2$-color universal have a different $2$-color 
universal set. For example, the balanced double-chain is not $2$-color universal for $K^+_{1,3}$ by 
Lemma~\ref{lem:forcat}, but it is easy to verify that the double-chain with one chain composed of only 
one vertex is.

There even exist graphs with a $2$-color universal set of points, but no double chain is $2$-color universal for 
them. Consider the properly colored $K^+_{1,4}$ with black central vertex. It is not embeddable on double-chains 
colored as in Fig.~\ref{f:k14-doublechain}. But a modification of the double chain in Fig.~\ref{f:glued-double-chain} 
is $2$-color universal for $K^+_{1,4}$.

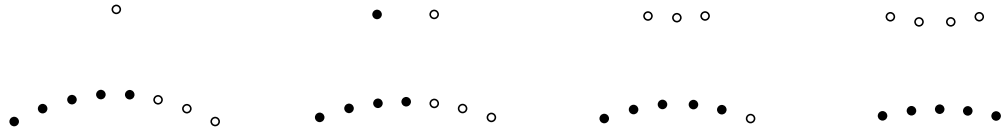
\myfig{k14-doublechain}{colorings of double-chains not admitting $K^+_{1,4}$}

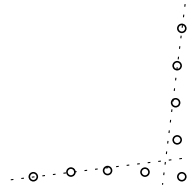
\myfig{glued-double-chain}{a $2$-color universal point set for $K^+_{1,4}$}

Some equitable bipartite planar graphs have no $2$-color universal set of points.

\begin{claim}
If $G$ is a bipartite planar quadrangulation on at least five vertices, then $G$ has no $2$-color universal 
set of points.
\end{claim}
\begin{proof}
Because the bipartite graph $G$ has no $3$-cycle, each of its faces has at least four vertices. 
Then, by Euler's formula, every planar drawing of $G$ is a quadrangulation.

Take a set $S$ of points and let $H(S)$ be the set of points of $S$ on its convex hull. In a straight-line 
planar drawing of a graph on a set $S$ of points, the points of $H(S)$ lie all on the outer face of the drawing. 
Thus $G$ can only be drawn on $S$ if $3 \leq |H(S)| \leq 4$. In a proper coloring of $G$ on at least five vertices,
one color class contains at least three vertices. If we color three points of $H(S)$ by this color and the rest
arbitrarily, $G$ cannot be embedded, because no face in a drawing of $G$ can contain three vertices of one color. 
\end{proof}

The results of this paper solve only a few particular cases of the following general question.
\begin{question}
Which planar bipartite graphs have a $2$-color universal set of points?
\end{question}

\section*{Acknowledgment} We thank Jakub \v{C}ern\'y for his active participation at
the earlier stages of our discussions.



\begin{thebibliography}{ABC}

\bibitem{abellanas99}
M. Abellanas, J. Garc{\'\i}a, G. Hernandez, M. Noy, P. Ramos, 
Bipartite embeddings of trees in the plane, 
Discrete Appl. Math. 93 (1999), 141--148.

\bibitem{abellanas03}
M. Abellanas, J. Garc{\'\i}a, F. Hurtado, J. Tejel, Caminos Alternantes (in Spanish),
Proc. X Encuentros de Geometr{\'\i}a Computacional, Sevilla, June 2003, pp.\ 7--12. English version
available on Ferran Hurtado's web page \url{http://www-ma2.upc.es/~hurtado/mypapers.html}.

\bibitem{aichholzer10}
O. Aichholzer, S. Cabello, R. Fabila-Monroy, D. Flores-Pe\~{n}aloza, T. Hackl, C. Huemer, F. Hurtado, D.R. Wood,
Edge-removal and non-crossing configurations in geometric graphs,
Discrete Math. Theor. Comput. Sci. 12 (2010), no. 1, 75--86.

%
\bibitem{brandes10+}
U. Brandes, C. Erten, A. Estrella-Balderrama, J. Fowler, F. Frati, M. Geyer, C. Gutwenger, Seok-Hee Hong, M. Kaufmann,
S. Kobourov, G. Liotta, P. Mutzel, A. Symvonis,
Colored Simultaneous Geometric Embeddings and Universal Pointsets,
Algorithmica, to appear, available online at \url{http://dx.doi.org/10.1007/s00453-010-9433-x}.


\bibitem{brassmoserpach}
P. Brass, W. Moser, J. Pach, Research Problems in Discrete Geometry,
Springer, New York, 2005.

\bibitem{cerny07}
J. \v{C}ern\'y, Z. Dvo\v{r}\'ak, V. Jel\'inek, J. K\'ara,
Noncrossing Hamiltonian paths in geometric graphs,
Discrete Appl. Math. 155 (2007), no. 9, 1096--1105.

\bibitem{gdversion}
J. Cibulka, J. Kyn\v{c}l, V. M\'{e}sz\'{a}ros, R. Stola\v{r}, P. Valtr,  
Hamiltonian alternating paths on bicolored double-chains,
in: I. G. Tollis, M. Patrignani (Eds.), Graph Drawing 2008, Lecture Notes in Computer Science 5417,
Springer, New York, 2009, pp.\ 181--192.

\bibitem{garcia}
A. Garc{\'\i}a, M. Noy, J. Tejel,
Lower bounds on the number of crossing-free subgraphs of $K_N$,
Comput. Geom. 16 (2000), 211--221.

\bibitem{giacomograph}
E. Di Giacomo, G. Liotta, F. Trotta,
On embedding a graph on two sets of points,
Int. J. of Foundations of Comp. Science 17(2006), no. 5, 1071--1094.

\bibitem{giacomopath}
E. Di Giacomo, G. Liotta, F. Trotta,
How to embed a path onto two sets of points,
in: P. Healy, N. Nikolov (Eds.), Graph Drawing 2005, Lecture Notes in Computer Science 3843,
Springer, New York, 2006, pp.\ 111--116.

\bibitem{gritzmann91}
P. Gritzmann, B. Mohar, J. Pach, R. Pollack,
Embedding a planar triangulation with vertices at specified points,
in: Am. Math. Monthly 98 (1991), 165--166 (Solution to problem E3341)

\bibitem{hajnalmeszaros}
P. Hajnal, V. M\'esz\'aros, 
Note on noncrossing path in colored convex sets, accepted to Discrete Math. Theor. Comput. Sci..

\bibitem{kanekokano}
A. Kaneko, M. Kano,
Discrete geometry on red and blue points in the plane --- a survey,
in: B. Aronov et al. (Eds.), Discrete and computational geometry,
The Goodman-Pollack Festschrift, Springer,
Algorithms Comb. 25 (2003), 551--570.

\bibitem{kanekokanosuzuki}
A. Kaneko, M. Kano, K. Suzuki, Path coverings of two sets of points in the plane,
in: J. Pach (Ed.), Towards a Theory of Geometric Graphs, Contemporary Mathematics 342
(2004), 99--111.

\bibitem{kaufmannwiese}
M. Kaufmann, R. Wiese,
Embedding vertices at points: Few bends suffice for planar graphs,
Journal of Graph Algorithms and Applications 6 (2002), no. 1, 115--129.

\bibitem{kynclpt}
J. Kyn\v cl, J. Pach, G. T\'oth, Long alternating paths in bicolored point sets, 
Discrete Mathematics 308 (2008), no. 19, 4315--4321.

\bibitem{pachwenger}
J. Pach, R. Wenger, 
Embedding planar graphs at fixed vertex locations, 
Graphs and Combinatorics 17 (2001), no. 4, 717--728.


\end{thebibliography}
\end{document}